\documentclass[letterpaper, 10 pt, conference]{ieeeconf}
\IEEEoverridecommandlockouts
\usepackage{hyperref}
\usepackage{graphics}
\usepackage{epsfig}
\usepackage{amsmath}
\usepackage{amssymb}
\usepackage{xcolor}
\let\labelindent\relax
\usepackage{enumitem}
\usepackage{multirow}
\usepackage{url}
\usepackage{lipsum}
\usepackage[noadjust]{cite} 
  
\usepackage{amsthm}
\theoremstyle{definition}
\newtheorem{lemma}{\normalfont \bfseries Lemma}
\newtheorem{corollary}{\normalfont \bfseries Corollary}

\newtheorem{theorem}{\normalfont\bfseries Theorem}
\newtheorem{assumption}{\normalfont\bfseries Assumption}

\newtheorem{definition}{\normalfont\bfseries Definition}

\newtheorem{remark}{\normalfont\bfseries Remark}

\newcommand{\Ns}{N}
\newcommand{\J}{\mathcal{J}}

\newcommand{\X}{\mathbb{R}^n}

\newcommand{\bx}{\mathbf{x}}

\newcommand{\bzero}{\mathbf{0}}
\renewcommand{\bf}{\mathbf{f}}
\newcommand{\bg}{\mathbf{g}}
\newcommand{\bu}{\mathbf{u}}
\newcommand{\bk}{\mathbf{k}}

\newcommand{\C}{\mathcal{C}}
\newcommand{\D}{\mathcal{D}}
\newcommand{\Cp}{\mathcal{C}_{\mathrm{p}}}

\newcommand{\Sb}{\mathcal{S}_\mathrm{b}}

\newcommand{\SIp}{\mathcal{S}_{\mathrm{I}}^{\mathrm{p}}}
\newcommand{\SI}{\mathcal{S}_{\mathrm{I}}}

\newcommand{\R}{\mathbb{R}}
\newcommand{\Rn}{\mathbb{R}^n}
\newcommand{\Rm}{\mathbb{R}^m}

\newcommand{\hb}{h_\mathrm{b}}
\newcommand{\kd}{\mathbf{k}_{\mathrm{d}}}
\newcommand{\kb}{\mathbf{k}_\mathrm{b}}

\newcommand{\U}{\mathcal{U}}

\newcommand{\K}{\mathcal{K}}
\newcommand{\Ke}{\K^{\rm e}}

\newcommand{\phib}{\boldsymbol{\varphi}_{\mathrm{b}}}

\newcommand{\ba}{\mathbf{a}}

\title{\LARGE \textbf{
Enforcing Mixed State-Input Constraints with Multiple Backup Control Barrier Functions: A Projection-based Approach
}}

\author{Laszlo Gacsi, Adam K. Kiss, Ersin Da\c{s}, and Tamas G. Molnar%
\thanks{L. Gacsi and T. G. Molnar are with the Department of Mechanical Engineering, Wichita State University, Wichita, KS 67260, USA,
{\tt\small lxgacsi@shockers.wichita.edu, tamas.molnar@wichita.edu}.}%
\thanks{Adam K. Kiss is with the Department of Applied Mechanics, Faculty of Mechanical Engineering, Budapest University of Technology and Economics, Budapest, Hungary, {\tt\small kiss\_a@mm.bme.hu}.}%
\thanks{Ersin Da\c{s} is with the Department of Mechanical, Materials, and Aerospace Engineering, Illinois Institute of Technology, Chicago, IL 60616, USA, {\tt\small edas2@illinoistech.edu}.}%
\vspace{-1mm}
}

\begin{document}

\maketitle
\thispagestyle{empty}
\pagestyle{empty}

\begin{abstract}
Ensuring the safety of control systems often requires the satisfaction of constraints on states (such as position or velocity), control inputs (such as force), and a mixture of states and inputs (such as power that depends on both velocity and force).
This paper presents a safety-critical control framework for enforcing {\em mixed state-input constraints} through a generalization of backup control barrier functions (backup CBFs).
First, we extend the backup CBF approach to maintain multiple {\em decoupled} state and input constraints using a single backup set--backup controller pair.
Second, we address {\em mixed} state-input constraints by converting them into state constraints using a projection from the state-input space to the state space along the backup controller.
In the special case of decoupled state and input constraints, the proposed method simplifies the synthesis of backup CBFs by eliminating the need for saturating backup control laws.
Finally, we demonstrate the efficacy of the proposed method on an inverted pendulum example, where constraints on the angle (state), torque (input), and power (mixture of state and input) are satisfied simultaneously.
\end{abstract}

\section{Introduction}
\label{sec:intro}

Safe behavior is an essential requirement for controlling modern autonomous systems.
Formally, safety is often defined as a state constraint--keeping the system's state within a safe set in the state space--which can be enforced, for example, using \textit{control barrier functions} (CBFs) \cite{ames2017cbf}.
Yet, this formulation often does not take into account the
limitations of control inputs explicitly.
Synthesizing safe controllers for dynamical systems subject to both state and input constraints is a common challenge in real-world control problems. 

In many applications, such as robotic systems and vehicles, state and input constraints are coupled.
For example, acceleration-based safety constraints for rigid-body robotic systems are both state and input-dependent \cite{de2019trajectory}.
For autonomous ground vehicles, the admissible stability region can vary with steering angle and braking force, making the corresponding safety constraint depend on both state and input \cite{huang2021stability}. Similarly, tipover avoidance constraints for legged and wheeled robots explicitly depend not only on states but also on inputs \cite{de2019trajectory, tan2025zero}. Input-dependent constraints also arise in passivity-based teleoperation of multi-robot systems \cite{notomista2024stable}, and the trajectory space forward invariance problem \cite{vahs2025forward}.

To address {\em decoupled} state and input limits within the  CBF framework, integral CBF methods \cite{ames2020integral, dacs2024rollover} augment the system dynamics by treating the inputs as new states.
Recently, in \cite{tan2025zero}, the notion of zero-order CBFs was introduced to enforce input-dependent safety via a one-step sampled-data condition; however, the resulting constraint is generally nonconvex.
An alternative approach is the
\textit{backup CBF method}~\cite{gurriet2020scalable,chen2021backup} which offers both formal safety guarantees and feasibility under input constraints.
Backup CBFs have been applied to mobile robots \cite{janwani2024learning}, spacecraft \cite{van2025safety},
manipulators \cite{Kobackup}, 
with other works exploring
non-smooth constraints \cite{he2025predictive}, timed reach-avoid specifications \cite{zhang2026safety}, and dynamic safety margin \cite{freire2025designing}. 
Yet, backup CBFs have not been established for {\em mixed} (coupled) state-input constraints.




\begin{figure}
    \centering
    \includegraphics[width=0.98\linewidth]{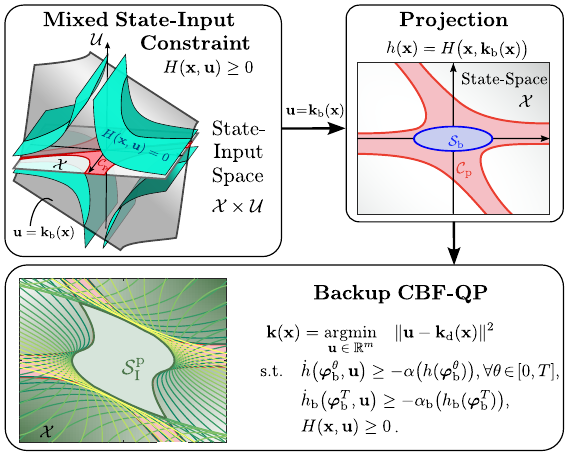}
    \vspace{-5mm}
    \caption{Overview of the proposed projection-based backup CBF approach.}
    \label{fig:overview}
    \vspace{-4mm}
\end{figure}

In this paper, we develop a safe control approach to enforce mixed state-input constraints using backup CBFs.
First, we extend the backup CBF method to handle {\em multiple decoupled} state and input constraints using a single backup set--backup controller pair.
Then, we address {\em mixed} constraints.
The key idea of our approach, shown in Fig.~\ref{fig:overview}, is a projection
from the state-input space to the state space, to transform mixed constraints into pure state constraints.
If the mixed constraints are affine in the input, the proposed backup CBF-based controller synthesis problem becomes an easy-to-solve quadratic program.
Moreover, for decoupled state and input constraints, our approach also helps to simplify the design of backup controllers and the calculation of backup CBFs by obviating the need for saturating backup controllers.
We showcase our approach on an inverted pendulum subject to angle (state), torque (input), and power (mixed) constraints.

\section{Input Constrained Safety}
\label{sec:CBF}
The system of interest is considered in the following form:
\begin{equation}
    \dot{\bx}=
    \bf(\bx)+\bg(\bx)\bu\,,
    \label{eq:control_system}
\end{equation}
with state
${\bx\in\X}$, input
${\bu\in\U\subseteq \Rm}$,
where $\U$ is assumed to be a convex polytope, and smooth functions
${\bf: \X \to \R^n}$,
${\bg: \X \to \R^{n\times m}}$. A locally Lipschitz controller
${\bk : \X \to \U}$,
${\bu = \bk(\bx)}$,
yields the closed-loop system:
\begin{equation}
    \dot{\bx}=\bf(\bx)+\bg(\bx)\bk(\bx)\,.
    \label{eq:closedloop}
\end{equation}
We revisit controllers that enforce state constraints, maintaining the state inside a constraint set, ${\bx(t) \in \C \subset \X}$, while respecting {\em decoupled} input limits, ${\bu(t) \in \U}$, for all time.

\subsection{Classical Control Barrier Function Framework}

The state constraint set $\C$ is defined using a smooth constraint  function ${\psi:\X\to\R}$:
\begin{align}
    \C&=\{ \bx\in \X : \psi(\bx)\geq0\}\,.
    \label{eq:constraint_set}
\end{align}
The safety of system~\eqref{eq:closedloop} is then synonymous to the forward invariance of $\C$: if ${\bx_0 \in \C}$, then ${\bx(t) \in \C}$ for all ${t \geq 0}$.
This can be ensured using control barrier functions (CBFs).

\begin{definition}[\cite{ames2017cbf}]
Function $\psi$ is a {\em control barrier function} for \eqref{eq:control_system} on $\C$ if there exists\footnote{Function  ${\alpha : [0,a) \to \R}$, ${a>0}$, is of class-$\K$ (${\alpha \in \K}$) if it is continuous, strictly increasing, and ${\alpha(0)=0}$.
Function ${\alpha : (-b,a) \to \R}$, ${a,b>0}$ is of extended class-$\K$ (${\alpha \in \Ke}$) if it has the same properties.} ${\alpha \in \Ke}$ such that:
\begin{equation}
    \sup_{\bu \in \U}\,\dot{\psi}(\bx,\bu) > -\alpha\big(\psi(\bx)\big)
    \label{eq:forward}
\end{equation}
for all ${\bx \in \C}$, with
${\dot{\psi}(\bx,\bu) = \nabla \psi(\bx) \cdot \big( \mathbf{f}(\bx)+\mathbf{g}(\bx)\bu \big)}$.
\end{definition} 
Given a CBF, the set $\C$ can be rendered forward invariant.

\begin{theorem}[\cite{ames2017cbf}]
    \textit{If $\psi$ is a CBF for \eqref{eq:control_system} on $\C$, then any locally Lipschitz controller $\mathbf{u = k(x)}$ satisfying:
\begin{equation}
    \dot{\psi}\big(\mathbf{x,k(x)}\big) \geq -\alpha\big(\psi(\bx)\big)\,,
    \label{eq:forward_inv}
\end{equation}
for all ${\bx \in \C}$, renders $\C$ forward invariant for~\eqref{eq:closedloop}.}
\label{th1}
\end{theorem}

This result can be used for safety-critical control synthesis, by solving the quadratic program (QP):
\begin{equation}
\begin{aligned}
\mathbf{k}(\bx) = \, \underset{\bu\in\U}{\operatorname{argmin}} \quad & \|\bu-\kd(\bx)\|^2\\
\textrm{s.t.} \quad & \dot{\psi}(\bx,\bu) \geq -\alpha\big(\psi(\bx)\big)\,,
\label{eq:QP}
\end{aligned}
\end{equation}
where the desired controller
${\kd : \X \to \U}$ is minimally modified in a way that $\C$ is rendered forward invariant. The CBF-QP in \eqref{eq:QP} is feasible if $\psi$ is a CBF as defined in \eqref{eq:forward}.

However, an arbitrary constraint function $\psi$ may not necessarily be a CBF.
Finding CBFs can be especially challenging with input constraints (${\U \subset \Rm}$), because the input that could maintain the system's safety might be outside the admissible input set.
A constructive approach that addresses both state and input constraints is the \textit{backup CBF method}.

\subsection{Backup Control Barrier Functions}
\label{sec:backup}

The foundation of the backup CBF method \cite{gurriet2020scalable,chen2021backup} is the use of a {\em backup controller} and an associated {\em backup set}.
The backup set is a subset of the constraint set, defined as:
\begin{equation}
    \Sb =\{ \bx\in \X : \hb(\bx)\geq0\}\,,
    \label{eq:backup_set}
\end{equation}
with the smooth function ${\hb : \X \to \R}$, such that ${\Sb \subseteq \C}$.
The backup controller is a smooth controller ${\kb : \X \to \U}$ rendering the backup set forward invariant along the corresponding closed-loop system:
\begin{equation}
    \dot{\bx}=\bf(\bx)+\bg(\bx)\kb(\bx)\,.
    \label{eq:f_b}
\end{equation}
The solution of \eqref{eq:f_b} over time ${\theta\geq 0}$ starting from state $\bx$ is the {\em backup flow} $\phib(\theta,\bx)$.
A procedure for constructing backup set--backup controller pairs is elaborated on in \cite{gacsi2025braking}.

Typically, the backup set is a small subset of the constraint set, representing safe but (very) conservative behavior.
Therefore, the backup controller $\kb$ is not used directly to keep the system inside the backup set $\Sb$, but rather to expand $\Sb$ into a larger set ${\SI\subseteq\C}$ that is invariant under $\kb$:
\begin{equation}
    \SI=
    \left\{ \bx \in \X :
    \begin{array}{l}
    \;\; \psi \big( \phib(\theta,\bx) \big) \geq 0, \ \forall \theta \in [0,T], \\
    \hb \big( \phib(T,\bx) \big) \geq 0
    \end{array}
    \! \! \right\}\,.
    \label{eq:SI}
\end{equation}
This set encodes the states from which the backup set can be safely reached over a horizon ${T>0}$ using the backup controller.
This construction enables the synthesis of controllers satisfying decoupled state and input constraints by rendering $\SI$ forward invariant, as established by the following theorem.

\begin{theorem}[\cite{molnar2023safety}] \label{theo:backup}
\textit{Consider the system \eqref{eq:control_system}, the set $\C$ in \eqref{eq:constraint_set}, a backup set ${\Sb \subseteq \C}$ in~\eqref{eq:backup_set}, a backup controller ${\kb : \X \to \U}$ that renders $\Sb$ forward invariant along~\eqref{eq:f_b}, and the set $\SI$ in \eqref{eq:SI} with ${T > 0}$.
Then, there exist ${\alpha,\alpha_\mathrm{b} \in \K}$ such that:
\begin{equation}
\begin{aligned}
    \dot{\psi}\big(\phib(\theta,\bx),\bu\big) & \!\geq\! -\alpha\big(\psi(\phib(\theta,\bx))\big),
    \ \forall \theta \!\in\! [0,T],\\
    \dot{h}_{\mathrm{b}} \big( \phib(T,\bx),\bu \big) & \!\geq\! -\alpha_{\mathrm{b}} \big( \hb(\phib(T,\bx)) \big)
\label{eq:backup_constr}
\end{aligned}
\end{equation}
holds for ${\bu = \kb(\bx)}$ for all ${\bx \in \SI}$.
Moreover, any locally Lipschitz controller ${\bk :\X \!\to\! \U}$, ${\bu = \bk(\bx)}$, that satisfies \eqref{eq:backup_constr} for all ${\bx \in \SI}$ renders ${\SI \subseteq \C}$ forward invariant along \eqref{eq:closedloop}.}
\label{th:2}
\end{theorem}

For control synthesis, the CBF-QP in~\eqref{eq:QP} is augmented with predictive constraints, leading to the {\em backup CBF-QP}:
\begin{equation}
\begin{aligned}
\mathbf{k}(\bx) & = \underset{{\bu\in\U} }{\operatorname{argmin}} \quad  \|\bu-\mathbf{k}_{\mathrm{d}}(\bx)\|^2 \\
\textrm{s.t.} \quad & \dot{\psi} \big( \phib(\theta,\bx),\bu \big) \geq -\alpha\big(\psi(\phib(\theta,\bx))\big)\,, \
\forall \theta \!\in\! [0,T]\,, \\
&\dot{h}_{\mathrm{b}} \big( \phib(T,\bx),\bu \big) \geq -\alpha_{\mathrm{b}}\big(\hb(\phib(T,\bx))\big)\,.
\label{eq:QP2}
\end{aligned}
\end{equation}
This QP is feasible even under input constraints with appropriately chosen $\alpha$, $\alpha_{\mathrm{b}}$ and it enforces state constraints (safety) on the basis of Theorem \ref{th:2}.
In practice, the constraints in the QP~\eqref{eq:QP2} are discretized into $N_{\rm c}$ constraints as they can only be evaluated at discrete $\theta$ time instants along the backup flow.
As for the backup flow, it can be computed via forward integration of the backup system~\eqref{eq:f_b} using an ordinary differential equation (ODE) solver or, when \eqref{eq:f_b} is linear, it can be solved analytically; see later in Remark~\ref{remark:linear}.

\section{Multiple Backup CBFs}
\label{sec:multiple}

Before addressing mixed state-input constraints, we extend the backup CBF method to handle multiple state constraints that are decoupled from input constraints.
Inspired by \cite{Kobackup}, we incorporate multiple constraints using a single backup set--backup controller pair.
This complements \cite{rabiee2025soft} where a single constraint and multiple backup controllers were used.
    
Hereby, we introduce $\Ns$ constraint functions  ${\psi_j: \X \to \R}$ for ${j \in \J=\{1, \dots, \Ns\}}$.
These define the constraint set: 
\begin{equation}
    \C = \big\{ \bx \in \X :\psi_j(\bx) \geq 0,\ \forall j \in \J \big\}.
    \label{eq:global_set}
\end{equation}
We consider a single backup set ${\Sb \subseteq \C}$ that respects all state constraints, and a backup controller ${\kb: \X \to \U}$ rendering it forward invariant.
Then, we redefine the invariant set in~\eqref{eq:SI}:
\begin{equation}
    \SI \!=\! \left\{ \bx \in \X : 
    \begin{array}{l}
    \psi_j \big( \phib(\theta, \bx) \big) \geq 0,\; \forall j \!\in\! \J,\; \forall \theta \!\in\! [0, T], \\
    h_{\rm b}\big(\phib(T, \bx)\big) \geq 0
    \end{array}
    \!\!\right\},
    \label{eq:SI_g}
\end{equation}
indicating the set of states from which the backup flow $\phib(\theta, \bx)$ satisfies all $\Ns$ state constraints for the duration ${\theta \in [0,T]}$ and terminates in $\Sb$ at ${\theta=T}$.
As a consequence of Theorem~\ref{theo:backup}, we can state the following corollary.

\begin{corollary}
\textit{Consider the system \eqref{eq:control_system}, the set $\C$ in \eqref{eq:global_set}, a backup set ${\Sb \subseteq \C}$ in~\eqref{eq:backup_set}, a backup controller ${\kb : \X \to \U}$ that renders $\Sb$ forward invariant along~\eqref{eq:f_b}, and the set $\SI$ in \eqref{eq:SI_g} with ${T > 0}$.
Then, there exist ${\alpha_j,\alpha_\mathrm{b} \in \K}$ such that:
\begin{equation}
\begin{aligned}
    \dot{\psi}_j\big(\phib(\theta, \bx), \bu\big) &\geq -\alpha_j\big(\psi_j(\phib(\theta, \bx))\big), \forall \theta \!\in\! [0, T], \forall j \!\in\! \J, \\
    \dot{h}_{\mathrm{b}}\big(\phib(T, \bx), \bu\big) &\geq -\alpha_{\mathrm{b}}\big(h_{\mathrm{b}}(\phib(T, \bx))\big)
\label{eq:backup_constr_g}
\end{aligned}
\end{equation}
holds for ${\bu = \kb(\bx)}$ for all ${\bx \in \SI}$.
Moreover, any locally Lipschitz controller ${\bk :\X \!\to\! \U}$, ${\bu = \bk(\bx)}$, that satisfies \eqref{eq:backup_constr_g} for all ${\bx \in \SI}$ renders ${\SI \subseteq \C}$ forward invariant along \eqref{eq:closedloop}.}
\label{th:3}
\end{corollary}

\begin{proof}
    For each individual constraint ${j \in \J}$, let us define the 
    constraint set
    ${\C_j = \{ \bx \in \X : \psi_j(\bx) \geq 0\}}$ and
    invariant set:
    \begin{equation}
        \SI^j = \left\{ \bx \!\in\! \X : 
        \begin{array}{l}
        \psi_j\big(\phib(\theta, \bx)\big) \geq 0,\ \forall \theta \in [0, T]\,, \\
        h_{\rm b}\big(\phib(T, \bx)\big) \geq 0
        \end{array}
        \!\!\right\}.
        \label{eq:SI_j}
    \end{equation}
    By definition of $\C$ in~\eqref{eq:global_set} and $\SI$ in~\eqref{eq:SI_g}, ${\bx\in\SI}$ if and only if
    ${\phib(\theta, \bx) \!\in\! \C_j}$,
    ${\forall \theta \!\in\! [0, T]}$,
    ${\forall j \!\in\! \J}$, while
    ${\phib(T, \bx) \!\in\! \Sb}$.
    This means ${\bx \in \SI \iff \bx \in \SI^j}$, ${\forall j \in \J}$,
    thus ${\SI = \bigcap_{j=1}^{\Ns} \SI^j}$.
    
    Based on Theorem~\ref{theo:backup}, for each individual constraint ${j \in \J}$ the backup controller ${\bu = \kb(\bx)}$ satisfies~\eqref{eq:backup_constr} with $\psi_j$ for all ${\bx \in \SI^j}$.
    Hence the backup controller $\bu = \kb(\bx)$ satisfies the combined constraint \eqref{eq:backup_constr_g} when ${\bx \in \SI}$.
    Note that this requires that the constraints share the same backup set and backup controller.
    Furthermore, based on Theorem~\ref{theo:backup}, for each constraint ${j \in \J}$ a locally Lipschitz controller $\bk$ satisfying~\eqref{eq:backup_constr} with $\psi_j$ for all ${\bx \in \SI^j}$ renders $\SI^j$ forward invariant.
    Because the intersection of forward invariant sets is forward invariant under the same control law, $\SI$ is forward invariant along \eqref{eq:closedloop} for controllers satisfying~\eqref{eq:backup_constr_g}.
    Finally, since ${\phib(0, \bx) = \bx}$, ${\bx \in \SI}$ implies ${\bx \in \C}$.
\end{proof}

\section{Mixed State-Input Constraints}
\label{sec:mixed}

In this section, we establish our main result: a projection-based framework to handle mixed state-input constraints. We first introduce this concept for a single constraint and then extend it to the multi-constraint case.

\subsection{Projection with Single Constraint}
\label{sec:mixed_single}

We define the mixed state-input constraint set $\D$ via the smooth function $H: \X \times \Rm \to \R$:
\begin{equation}
    \D = \big\{ (\bx, \bu) \in \X \times \Rm : H(\bx, \bu) \geq 0 \big\}.
    \label{eq:mixed_set}
\end{equation}
Our objective is to find a locally Lipschitz controller $\bk$ such that the closed-loop system~\eqref{eq:closedloop} satisfies
${\big(\bx(t), \bk(\bx(t))\big) \in \D}$,
or equivalently
${H \big( \bx(t), \bk(\bx(t)) \big) \geq 0}$,
for all ${t \geq 0}$
(at least for a certain set of initial conditions to be discussed later).
For now, we focus on a single mixed constraint and omit other state, input, or mixed constraints---these will be addressed in the next subsection. 
Unlike pure state constraints, $H(\bx, \bu)$ has \textit{relative degree zero} because the input $\bu$ appears explicitly. While one could be tempted to directly enforce the constraint ${H(\bx, \bu) \geq 0}$ 
in the optimization, this may not be recursively feasible. To ensure recursive feasibility for mixed state-input constraints, we extend the backup CBF method.

Our key idea is to project the constraint from the state-input space ${\X \times \Rm}$ to the state space $\X$ along the backup controller ${\bu = \kb(\bx)}$.
We define the {\em projected constraint set}:
\begin{equation}
    \Cp = \big\{ \bx \in \X : (\bx, \kb(\bx)) \in \D \big\} = \{ \bx \in \X : h(\bx) \geq 0 \},
    \label{eq:projected_set}
\end{equation}
with the {\em projected constraint function} ${h : \X \to \R}$:
\begin{equation}
    h(\bx) \triangleq H\big(\bx, \kb(\bx)\big)\,,
    \label{eq:h_p}
\end{equation}
which is smooth because both $H$ and $\kb$ are smooth.

We require the following properties for the backup set--backup controller pair.
\begin{assumption} \label{assum:backup}
    The backup set satisfies ${\Sb \subseteq \Cp}$, while the backup controller ${\kb : \X \to \Rm}$ renders the backup set $\Sb$ forward invariant along \eqref{eq:f_b}.
\end{assumption}
Then, using the backup flow $\phib$ corresponding to \eqref{eq:f_b}, we define the invariant set:
\begin{equation}
    \SIp = \left\{ \bx \in \Rn : 
    \begin{array}{l}
    h\big(\phib(\theta, \bx)\big) \geq 0, \ \forall \theta \in [0, T], \\
    h_{\rm b}\big(\phib(T, \bx)\big) \geq 0
    \end{array}
    \right\}.
    \label{eq:SI_Cp}
\end{equation}
We establish that the mixed state-input constraint is satisfied when using the backup controller inside this set.

\begin{lemma} \label{lemma:1}
\textit{
If $\bx \in \SI^{\rm p}$, then $\big(\bx, \kb(\bx)\big) \in \D$.
}
\end{lemma}
\begin{proof}
By the definition of $\SIp$ in \eqref{eq:SI_Cp}, ${\bx \in \SIp}$ implies ${h\big(\phib(0, \bx)\big) \geq 0}$. Since the backup flow at ${\theta=0}$ is ${\phib(0,\bx) = \bx}$, it follows that $h(\bx) \geq 0$. By the definition of $h(\bx)$ in \eqref{eq:h_p}, this is equivalent to $\big(\bx, \kb(\bx)\big) \in \D$.
\end{proof}

\begin{remark}
It is important to highlight that $\bx \in \SI^{\rm p}$ does not inherently yield ${\big(\bx, \bk(\bx)\big) \in \D}$ for an arbitrary controller $\bk$.
Thus, keeping $\bx$ inside $\SIp$ does not guarantee safety on its own.
However, $\SIp$ helps to establish recursive feasibility because $\kb$ is a controller that renders $\SIp$ forward invariant.
\end{remark}

\begin{lemma} \label{lemma:2}
\textit{
The backup controller $\kb$ renders the set $\SI^{\rm p}$ in \eqref{eq:SI_Cp} forward invariant along~\eqref{eq:f_b}.
}
\end{lemma}
\begin{proof}
Let ${\bx_0 \in \SI^{\rm p}}$.
To show forward invariance, we establish that
${\phib(t,\bx_0) \in \SI^{\rm p}}$ for all ${t \geq 0}$,
which is equivalent to\footnote{Observe that the flow satisfies
${\phib\big(\theta, \phib(t, \bx_0)\big) = \phib(t+\theta, \bx_0)}$  for any ${t \geq 0}$ and ${\theta \in [0,T]}$.}
${\phib(t+\theta,\bx_0) \in \Cp}$
and
${\phib(t+T,\bx_0) \in \Sb}$
for all ${t \geq 0}$ and ${\theta \in [0,T]}$.
%
%
Because ${\bx_0 \in \SI^{\rm p}}$, we have
${\phib(t+\theta, \bx_0) \in \Cp}$ for ${t + \theta \leq T}$.
Furthermore, because ${\bx_0 \in \SI^{\rm p}}$, we also have that ${\phib(T, \bx_0) \in \Sb}$.
Since $\Sb$ is forward invariant under $\kb$, this means 
${\phib(t+T, \bx_0) \in \Sb}$ for all ${t \geq 0}$, and also
${\phib(t+\theta, \bx_0) \in \Sb \subseteq \Cp}$ for any ${t+\theta > T}$.
\end{proof}



We now link the forward invariance of $\SI^{\rm p}$ to the corresponding conditions satisfied by the backup control input.
\begin{lemma} \label{lemma:3}
\textit{
There exist ${\alpha, \alpha_{\rm b} \in \K}$ such that ${\bu = \kb(\bx)}$ satisfies the following constraints for all ${\bx \in \SI^{\rm p}}$:
\begin{align}
\begin{aligned}
    \dot{h}\big(\phib(\theta, \bx), \bu\big) &\geq -\alpha\big(h(\phib(\theta, \bx))\big), \quad \forall \theta \in [0, T]\,, \\
    \dot{h}_{\rm b}\big(\phib(T, \bx), \bu\big) &\geq -\alpha_{\rm b}\big(h_{\rm b}(\phib(T, \bx))\big).
    \end{aligned}
\end{align}
}
\end{lemma}
\begin{proof}
This is a direct consequence of Lemma~\ref{lemma:2} and the classical backup CBF result in Theorem~\ref{theo:backup}, by replacing $\psi$ with $h$.
The proof follows that of \cite[Lemma~2]{molnar2023safety}.
\end{proof}

These lemmas lead to our main theorem that states conditions for the control input to ensure safety with respect to mixed state-input constraints in a recursively feasible way.

\begin{theorem} \label{th:4}
\textit{
Consider the system~\eqref{eq:control_system},
the set $\D$ in \eqref{eq:mixed_set}, a backup set $\Sb$ in~\eqref{eq:backup_set}, a backup controller ${\kb : \X \to \Rm}$, the projected constraint set $\Cp$ in \eqref{eq:projected_set}, and the set $\SI^{\rm p}$ in \eqref{eq:SI_Cp} with ${T > 0}$.
Let Assumption~\ref{assum:backup} hold.
Then, there exist $\alpha, \alpha_{\rm b} \in \Ke$ such that:
\begin{subequations} \label{eq:cbf_constraints_mixed}
\begin{align}
    \dot{h}\big(\phib(\theta, \bx), \bu\big) &\geq -\alpha\big(h(\phib(\theta, \bx))\big)\,,\ \forall \theta \in [0, T], \label{eq:cbf_flow} \\
    \dot{h}_{\rm b}\big(\phib(T, \bx), \bu\big) &\geq -\alpha_{\rm b}\big(h_{\rm b}(\phib(T, \bx))\big)\,, \label{eq:cbf_term} \\
    H(\bx, \bu) &\geq 0 \label{eq:cbf_mixed}
\end{align}
\end{subequations}
holds for ${\bu = \kb(\bx)}$ for all ${\bx \in \SI^{\rm p}}$.
Moreover, any locally Lipschitz controller ${\bk :\X \to \Rm}$, ${\bu = \bk(\bx)}$, that satisfies \eqref{eq:cbf_constraints_mixed} for all ${\bx \in \SIp}$ renders $\SIp$ forward invariant along \eqref{eq:closedloop}, implying ${\big(\bx(t), \bk(\bx(t))\big) \in \D}$ for all ${t \geq 0}$ if ${\bx_0 \in \SIp}$.
}
\end{theorem}
\begin{proof}
By Lemma \ref{lemma:3}, ${\bu = \kb(\bx)}$ satisfies \eqref{eq:cbf_flow} and \eqref{eq:cbf_term}, and by Lemma \ref{lemma:1}, ${\bu = \kb(\bx)}$ satisfies \eqref{eq:cbf_mixed}. 
For ${\bu = \bk(\bx)}$, the constraints in \eqref{eq:cbf_flow} and \eqref{eq:cbf_term} ensure the forward invariance of $\SI^{\rm p}$ based on Theorem~\ref{theo:backup}.
This means that ${\bx_0 \in \SIp}$ implies ${\bx(t) \in \SIp}$ for all ${t \geq 0}$.
Finally, the constraint in~\eqref{eq:cbf_mixed} directly enforces that the mixed state-input safety condition $\big(\bx(t), \bk(\bx(t))\big) \in \D$ is satisfied $\forall t\geq 0$.
\end{proof}

\begin{remark}
The constraints~\eqref{eq:cbf_flow} and~\eqref{eq:cbf_term} are required for guaranteeing {\em recursive feasibility}.
They enforce
${\bx(t) \in \SIp}$,
thereby restricting the system to visit only those states where there exists at least one controller (the backup controller) satisfying the mixed constraints.
The constraint~\eqref{eq:cbf_mixed} is required for {\em safety}, that is, for satisfying the mixed constraints with the chosen control law, such that ${\big( \bx(t), \bk(\bx(t)) \big) \in \D}$.
\end{remark}

\subsection{Projection with Multiple Constraints}
By combining the ideas from Section~\ref{sec:multiple} and Section~\ref{sec:mixed_single}, we extend the projection method to multiple constraints:
\begin{equation}
    \D = \big\{ (\bx, \bu) \in \X \times \Rm : H_j(\bx, \bu) \geq 0,\ \forall j \in \J \big\}\,,
\end{equation}
with smooth constraint functions  ${H_j: \X \times \Rm \to \R}$.
Note that this involves state constraint (${\D = \C \times \Rm}$), input constraint (${\D = \Rn \times \U}$), and the combination of decoupled state and input constraints (${\D = \C \times \U}$) as special cases, while also capturing general mixed constraints in the state-input space.

We propose to use the projected constraint set:
\begin{equation}
\begin{aligned}
    \Cp & = \big\{ \bx \in \X : \big(\bx, \kb(\bx)\big) \in \D \big\} \\
    & = \{ \bx \in \X : h_j(\bx) \geq 0,\ \forall j\in\J \},
    \label{eq:projected_set_multi}
\end{aligned}
\end{equation}
with the projected constraint functions ${h_j : \X \to \R}$:
\begin{equation}
    h_j(\bx) \triangleq H_j\big(\bx, \kb(\bx)\big).
    \label{eq:h_p_multi}
\end{equation}
Then, with a backup set $\Sb$ and a backup controller $\kb$ satisfying Assumption~\ref{assum:backup}, we construct the invariant set:
\begin{equation}
    \SIp = \left\{ \bx \in \Rn : 
    \begin{array}{l}
    h_j\big(\phib(\theta, \bx)\big) \geq 0,  \forall \theta \!\in\! [0, T],\! \forall j \!\in\! \J,\!\\
    h_{\rm b}\big(\phib(T, \bx)\big) \geq 0
    \end{array}
    \!\!\right\}\!,
\end{equation}
see Fig.~\ref{fig:overview} (where ${h_j\big(\phib(\theta,\bx) \big) \!=\! 0}$ is plotted for various $\theta$).
 
Based on the results of the previous sections, we enforce the invariance of this set and the satisfaction of the mixed constraints using \eqref{eq:cbf_flow} for each $h_j$, \eqref{eq:cbf_term}, and \eqref{eq:cbf_mixed} for each $H_j$.
Thus, we propose the optimization problem (OP):
\begin{equation} \label{eq:QP3}
\begin{aligned}
\mathbf{k}(\bx) = \underset{\bu \in \Rm}{\operatorname{argmin}} \quad \|\bu-\mathbf{k}_{\mathrm{d}}& (\bx)\|^2 \\
\textrm{s.t.} \quad \dot{h}_j \big( \phib(\theta,\bx),\bu \big) & \geq -\alpha_j\big(h_j(\phib(\theta,\bx))\big),
\\ & \qquad \forall \theta \!\in\! [0,T]\,, \forall j\in\J\,, \\
\dot{h}_{\mathrm{b}} \big( \phib(T,\bx),\bu \big) & \geq -\alpha_{\mathrm{b}}\big(\hb(\phib(T,\bx))\big)\,,\\
H_j(\bx,\bu) & \geq 0\,, \quad \forall j\in\J\,.
\end{aligned}
\end{equation}
By construction, this OP is recursively feasible because ${\bk(\bx) = \kb(\bx)}$ is a feasible solution.

\begin{corollary}
\textit{
There exist ${\alpha_j, \alpha_{\rm b} \in \Ke}$ such that the backup CBF-OP~\eqref{eq:QP3} is feasible for all ${\bx \in \SIp}$.
Furthermore, the controller~\eqref{eq:QP3} ensures for all ${\bx_0 \in \SIp}$ that ${\bx(t) \in \SIp}$ and ${\big(\bx(t), \bk(\bx(t))\big) \in \D}$ hold for all $t \geq 0$ for system~\eqref{eq:closedloop}.
}
\end{corollary}

\begin{remark}
The backup CBF-OP~\eqref{eq:QP3} is a quadratic program if and only if $H_j$ are affine in $\bu$:
\begin{equation}
    H_j(\bx, \bu) = b_j(\bx) - \ba_j(\bx)\bu\,, \quad \forall j \in \J \,,
    \label{eq:affine_H}
\end{equation}
where ${\ba_j : \X \to \R^{1 \times m}}$, ${b_j : \X \to \R}$ are smooth functions.
\end{remark}

\begin{remark}
For state constraints (${\ba_j(\bx)=\bzero}$), the last constraint with $H_j$ in~\eqref{eq:QP3} is inactive and can be omitted.
For input constraints ($\ba_j$ and $b_j$ are constants), $H_j$ encodes the input constraint and the optimization is done over ${\bu \in \Rm}$, rather than ${\bu\in\U}$ as in~\eqref{eq:QP2}, while the constraint with $h_j$ is also added.
Hence, for the special case of decoupled state and input constraints (discussed in Section~\ref{sec:backup}) we obtain:
\begin{equation}
\begin{aligned}
\mathbf{k}(\bx) = \underset{\bu \in \Rm}{\operatorname{argmin}} \quad & \|\bu-\mathbf{k}_{\mathrm{d}}(\bx)\|^2 \\
\textrm{s.t.} \quad \dot{\psi} \big( \phib(\theta,\bx),\bu \big) & \geq -\alpha\big(\psi(\phib(\theta,\bx))\big)\,,\ \forall \theta \in [0,T]\,, \\
- \ba_j \dot{\bk}_{\rm b}(\bx,\bu) & \geq -\alpha_j\big(b_j - \ba_j \kb(\bx) \big) \,,\ \forall j\in\J\,, \\
\dot{h}_{\mathrm{b}} \big( \phib(T,\bx),\bu \big) & \geq -\alpha_{\mathrm{b}}\big(\hb(\phib(T,\bx))\big)\,,\\
b_j - \ba_j \bu & \geq 0\,, \quad \forall j\in\J\,.
\end{aligned}
\end{equation}
Compared to the standard backup CBF-QP~\eqref{eq:QP2}, this includes the additional second constraint with the derivative of $\kb$.
\end{remark}

\begin{remark}
The projection may help to simplify the backup controller design.
For decoupled state and input constraints (${\D = \C \times \U}$), the standard backup CBF method requires 
${\Sb \subseteq \C}$ and ${\kb(\bx) \in \U}$ for all ${\bx \in \X}$.
This often means that the backup controller needs to be saturated to enforce input constraints over the entire state space.
As opposed, the projection-based method requires ${\Sb \subseteq \Cp}$, meaning ${\Sb \subseteq \C}$ and ${\kb(\bx) \in \U}$ for all ${\bx \in \Sb}$.
So input constraints only need to hold inside the (small) backup set $\Sb$, eliminating the need for saturation.
This can also simplify the computation of the backup flow.
For example, if~\eqref{eq:control_system} is feedback linearizable, the backup controller can be constructed so that the backup system~\eqref{eq:f_b} is linear, ${\dot{\bx}=\mathbf{A}\bx}$.
Then, the backup flow:
\begin{equation}
    \phib(\theta, \bx) = e^{\mathbf{A}\theta}\bx
    \label{eq:analytical_flow}
\end{equation}
can be calculated analytically.
This eliminates the need for forward integration at each time step, significantly reducing the computational cost of the backup CBF method.
The price of these benefits is that the projection reduces the size of the set where the system operates (i.e., ${\Cp \subseteq \C}$).
\label{remark:linear}
\end{remark}

\section{State-Input-Power Constrained Inverted Pendulum}
We demonstrate the proposed method on an inverted pendulum with unit-coefficient equations of motion:
\begin{equation}
    \underbrace{\begin{bmatrix}
        \dot{x}_1 \\
        \dot{x}_2
    \end{bmatrix}}_{\dot{\bx}} =
    \underbrace{\begin{bmatrix}
        x_2 \\
        \sin(x_1)
    \end{bmatrix}}_{\bf(\bx)} +
    \underbrace{\begin{bmatrix}
        0 \\ 1
    \end{bmatrix}}_{\bg(\bx)} u,
    \label{eq:pendulum}
\end{equation}
where the state ${\bx=[x_1~ x_2]^{\top}}$ contains angle and angular velocity, and the input $u$ is a torque, with power ${P = x_2 u}$.
The primary control objective is to keep the angle $x_1$ within the  bounds ${|x_1|\leq\varphi_{\rm max}}$ with the desired controller ${k_{\rm d}(\bx)=0}$. Moreover, we require input constraints ${u_{\rm min}\leq u\leq u_{\rm max}}$ and power constraints ${P_{\rm min}\leq x_2 u \leq P_{\rm max}}$.
These mixed constraints can be expressed in the affine form~\eqref{eq:affine_H} with:
\begin{align}
    H_1(\bx,\bu) &= \varphi_{\max} - x_1\,, \quad
    H_2(\bx,\bu) = x_1 + \varphi_{\max} \,,
    \label{eq:constr_pos} \\
    H_3(\bx,\bu) &= u_{\max} - u\,, \quad\;\;
    H_4(\bx,\bu) = u-u_{\min}\,,
    \label{eq:constr_u} \\
    H_5(\bx,\bu) &= P_{\max} - x_2 u\,, \;\;
    H_6(\bx,\bu) = x_2 u - P_{\min}\,.
    \label{eq:constr_pwr}
\end{align}
The parameters of this example are listed in Table~\ref{tab:parameters}.

\begin{table}
\centering
\caption{Parameters of the inverted pendulum example}
\vspace{-2mm}
\label{tab:parameters}
\setlength{\tabcolsep}{2.8pt} 
\small 
\begin{tabular}{|c|c|c||c|c|c||c|c|c|} 
\hline
Par. & Val. & Unit & Par. & Val. & Unit & Par. & Val. & Unit \\ \hline
$u_{\rm min}$ & $-1.1$ & Nm & $\varphi_{\rm max}$ & $1.75$ & rad & $\alpha_j(h_j)$ & $15 h_j$ & 1/s \\  
$u_{\rm max}$ & $1.2$  & Nm & $K$& $0.7$  & Nms & $\alpha_{\rm b}(\hb)$ & $\hb$ & 1/s \\
$P_{\rm min}$ & $-0.7$ & W & $X_2$ & $0.2$  & rad/s & $T$ & $8$ & s \\  
$P_{\rm max}$ & $0.2$  & W & $\gamma$ & $0.7$  & 1/s & $N_{\rm c}$ & $100$ & 1 \\ \hline
\end{tabular}
\end{table}

For the backup controller, we utilize feedback linearization and velocity feedback with gain ${K>0}$, as in \cite{Kobackup}:
\begin{equation}
    k_{\rm b}(\bx)=-\sin(x_1) -K x_2\,.
    \label{eq:pend_backup_controller}
\end{equation}
The six mixed state-input constraints are then projected onto the state space,
creating the projected constraint set $\Cp$ in~\eqref{eq:projected_set_multi}.
Note that input limits are captured by $\Cp$, thus it is not necessary to saturate the backup controller.
Under the backup policy~\eqref{eq:pend_backup_controller}, the system becomes linear and the backup flow can be computed analytically using~\eqref{eq:analytical_flow}:
\begin{equation}
    \phib(\theta,\bx)=\begin{bmatrix}
        \varphi_1(\theta,\bx)\\
        \varphi_2(\theta,\bx)
    \end{bmatrix}=\begin{bmatrix}
        x_1+\frac{1}{K}(1-e^{-K \theta}) x_2\\
        e^{-K\theta}x_2
    \end{bmatrix}\,,
    \label{eq:pend_flow}
\end{equation}
thus there is no need to integrate the backup flow online, and backup CBFs become explicit instead of implicit.
The backup flow converges to the equilibrium ${\bx^*=[x_1 + \frac{x_2}{K}~ 0]^{\top}}$, whose location depends on the initial condition. 
The backup set is then chosen to be a subset of $\Cp$, such that it contains all possible equilibria inside $\Cp$. One such backup set is an ellipse of size ${(\varphi_{\rm max},X_2)}$, given by:
\begin{equation}
    h_{\rm b}(\bx)=1-\left(\frac{x_1}{\varphi_{\rm max}}\right)^2-\left(\frac{x_2}{X_2}\right)^2\,. 
    \label{eq:pend_backup_set}
\end{equation} 

To highlight the effect of input and power constraints, we will compare our method to a classical, high-order CBF (HOCBF) technique \cite{xiao2022hocbf} that only enforces state limits via:
\begin{equation}
    \psi_{\rm e}(\bx)=-2x_1x_2+\gamma \big(\varphi_{\rm max}^2-x_1^2\big)\,,
    \label{eq:pend_hocbf}
\end{equation}
with ${\gamma > 0}$, derived from ${\psi(\bx) = \varphi_{\rm max}^2-x_1^2}$.

\begin{figure}
    \centering
    \includegraphics[width=0.99\linewidth]{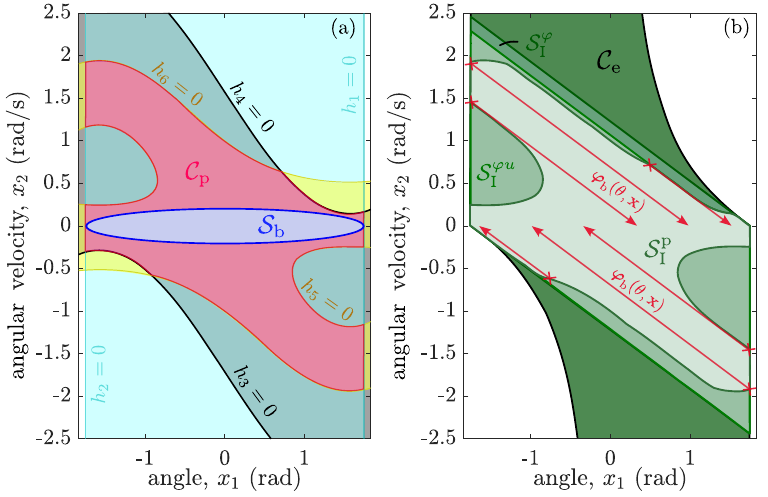}
    \vspace{-7mm}
    \caption{
    Set analysis for the inverted pendulum~\eqref{eq:pendulum} with mixed state-input constraints.
    Panel (a) depicts the six constraint sets from~\eqref{eq:constr_pos}-\eqref{eq:constr_pwr} that are projected using the backup controller~\eqref{eq:pend_backup_controller}
    (light blue: state, gray: input, yellow: power constraints).
    Their intersection, $\Cp$ from~\eqref{eq:projected_set_multi} (red), contains the backup set $\Sb$ from~\eqref{eq:pend_backup_set} (blue).
    Panel (b) compares four sets:
    the zero-superlevel set $\C_{\rm e}$ of the HOCBF in~\eqref{eq:pend_hocbf},
    the invariant set $\SI^{\varphi}$ for state constraints~\eqref{eq:constr_pos},
    the set $\SI^{\varphi u}$ for state and input constraints~\eqref{eq:constr_pos}-\eqref{eq:constr_u},
    and the set $\SIp$ for all six constraints~\eqref{eq:constr_pos}-\eqref{eq:constr_pwr}.
    The backup flows~\eqref{eq:pend_flow} launched from different initial conditions are also highlighted (red lines).}
    \vspace{-5mm}
\label{fig:sets}
\end{figure}

Figure~\ref{fig:sets} shows the sets $\Cp$, $\Sb$, $\SIp$ for the proposed method.
Panel (a) depicts the projected constraint set $\Cp$ (red) as the intersection of the zero-superlevel sets of
$h_j$ obtained
from~\eqref{eq:constr_pos}-\eqref{eq:constr_pwr} (light blue, gray, and yellow areas show state, input, and power constraints, respectively).
The backup set $\Sb$ (blue) from~\eqref{eq:pend_backup_set} lies inside $\Cp$ for the chosen parameters.
Panel (b) depicts the corresponding invariant sets.
The zero-superlevel set of the HOCBF $\psi_{\rm e}$ is $\C_{\rm e}$.
The invariant set $\SI^{\varphi}$ is calculated by considering only the state constraints~\eqref{eq:constr_pos}, $\SI^{\varphi u}$ is based on state and input constraints~\eqref{eq:constr_pos}-\eqref{eq:constr_u}, while $\SIp$ accounts for all six constraints~\eqref{eq:constr_pos}-\eqref{eq:constr_pwr}.
Notice how the sets are getting smaller.
The smaller size of $\SI^{\varphi}$ compared to $\C_{\rm e}$ indicates the effect of using backup CBFs to ensure recursive feasibility.
The size of ${\SI^{\varphi u}}$ shows the effect of input constraints, while $\SIp$ highlights the impact of power constraints.
All of these sets were computed analytically, as the backup flow~\eqref{eq:pend_flow} is explicit.
The backup flow (plotted in red) exponentially approaches the horizontal axis along oblique lines of slope $-K$ without exiting $\SI^{\varphi}$, $\SI^{\varphi u}$, or $\SIp$. 

Figure~\ref{fig:simulation} depicts the numerical simulation of~\eqref{eq:pendulum} using the backup CBF-QP~\eqref{eq:QP3} with constraints~\eqref{eq:constr_pos}-\eqref{eq:constr_pwr}, backup controller~\eqref{eq:pend_backup_controller}, backup set~\eqref{eq:pend_backup_set}, and two initial conditions (solid blue and orange).
We compare this to the HOCBF method using~\eqref{eq:pend_hocbf} in the QP~\eqref{eq:QP} while including input and power constraints directly in the QP.
The HOCBF-QP becomes infeasible, hence the trajectories (dotted blue and orange) leave $\C_{\rm e}$ slightly and violate the input and power limits.
Meanwhile, the backup CBF-QP satisfies all mixed state-input constraints, even for this open-loop unstable system.
This is thanks to forward prediction, which makes the backup CBF intervene earlier than the HOCBF; see panel (c).

\begin{figure}
    \centering
    \includegraphics[width=0.99\linewidth]{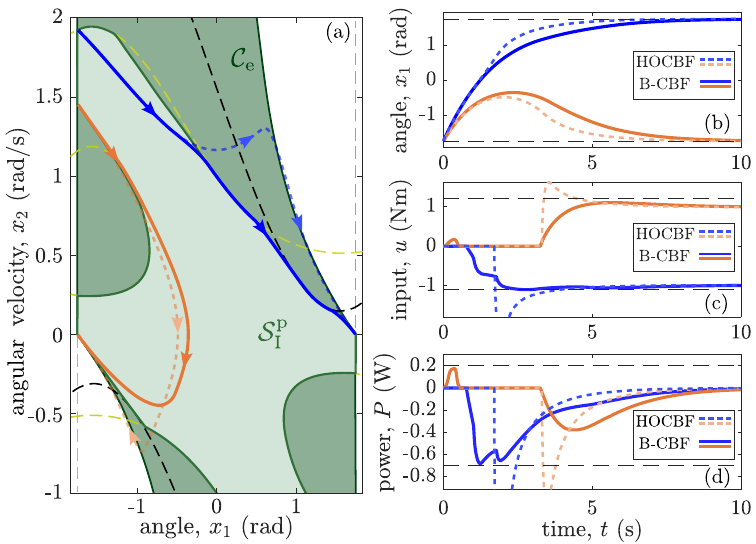}
    \vspace{-7mm}
    \caption{Simulation of the inverted pendulum~\eqref{eq:pendulum} using the proposed backup CBF-QP~\eqref{eq:QP3} (solid) and, for comparison, the HOCBF~\eqref{eq:pend_hocbf} (dotted).
    Panel (a) illustrates two trajectories (blue and orange) initiated from the invariant set $\SIp$ (green) where they remain throughout the simulation.
    The boundaries of the projected constraint set $\Cp$
    are indicated by dashed lines (light blue: state, black: input, yellow: power constraint).
    Panels (b)-(d) display the evolution of angle, input, and power signals, with dashed black lines representing their limits.
    The proposed method maintains all constraints.
    }
    \vspace{-5mm}
\label{fig:simulation}
\end{figure}

\section{Conclusion}
\label{sec:concl}

In this paper, we proposed a systematic control approach to enforcing mixed state-input constraints.
We introduced a projection to convert input constraints and more complex, mixed (e.g., power) constraints to state constraints, and enforce them using backup CBFs in a recursively feasible manner.
Our approach also enhances the classical backup CBF framework by eliminating the need for saturated backup controllers and potentially enabling analytical backup flow calculation.
We demonstrated our theoretical results on an inverted pendulum with state, input, and power constraints.





\bibliographystyle{IEEEtran}
\bibliography{lcss_cdc_2026}	

\end{document}